\documentclass[sn-mathphys-num]{sn-jnl}

\usepackage{graphicx}
\usepackage{multirow}
\usepackage{amsmath,amssymb,amsfonts}
\usepackage{amsthm}
\usepackage{mathrsfs}
\usepackage[title]{appendix}
\usepackage{xcolor}
\usepackage{textcomp}
\usepackage{manyfoot}
\usepackage{booktabs}
\usepackage{algorithm}
\usepackage{algorithmicx}
\usepackage{algpseudocode}
\usepackage{listings}

\usepackage{scalerel}
\usepackage{stackengine}
\usepackage{tikz}
\usepackage{enumitem}
\usepackage{stmaryrd}  
\usepackage{xparse}

\usetikzlibrary{
    arrows,
    arrows.meta,
    backgrounds,
    calc,
    decorations,
    decorations.markings,
    decorations.pathreplacing,
    decorations.text,
    fadings,
    fit,
    intersections,
    matrix,
    patterns,
    positioning,
    shadows.blur,
    shapes,
    shapes.geometric,
    trees,
}

\makeatletter
\def\thm@space@setup{
  \thm@preskip=3pt plus 2pt minus 1pt
  \thm@postskip=\thm@preskip
}
\makeatother

\theoremstyle{thmstyleone}

\newtheorem{lemma}{Lemma}[section]
\newtheorem{proposition}[lemma]{Proposition}

\newtheorem{theorem}[lemma]{Theorem}

\theoremstyle{thmstyletwo}

\theoremstyle{thmstylethree}

\theoremstyle{definition}
\newtheorem{definition}[lemma]{Definition}
\newtheorem{remark}[lemma]{Remark}

\NewDocumentEnvironment{axioms}{O{(Ax-\expandafter\arabic*)}}{\enumerate[label=#1, labelindent=3em, leftmargin=6em, itemindent=2em, labelsep=2em]}{\endenumerate}

\newcommand\ee[1]{\enspace #1 \enspace}
\newcommand\ete[1]{\enspace\text{#1}\enspace}
\newcommand\qtq[1]{\quad\text{#1}\quad}

\newcommand\p[1]{\mathcal{#1}}
\newcommand\ol[1]{\overline{#1}}
\newcommand\sue{\subseteq}
\newcommand{\id}{\ensuremath{\text{id}}}

\newcommand\theStar[1]{\textbf{($\star #1$)}}
\newcommand\addStar[1]{\hfill\theStar{#1}\quad\qquad~}

\newcommand\twist{^{\twistSym}}
\newcommand\twistSym{{\bowtie}}

\newcommand\oimply{\supset}
\newcommand\oimplied{\subset}
\newcommand\simply{\rightarrowtriangle}

\newcommand{\dtt}{\ensuremath{\text{\emph{t}}\mkern-3mu\text{\emph{t}}}}
\newcommand{\dff}{\ensuremath{\text{\emph{f}}\mkern-3mu\text{\emph{f}}}}
\newcommand{\mtt}{\ensuremath{\mathsf{t}}}
\newcommand{\mff}{\ensuremath{\mathsf{f}}}
\newcommand{\FOUR}{\ensuremath{\mathcal{FOUR}}}
\newcommand{\two}{\ensuremath{\mathbf{2}}}

\newcommand{\tot}{\ensuremath{\mathsf{tot}}}
\newcommand{\con}{\ensuremath{\mathsf{con}}}

\newcommand{\ileq}{\sqsubseteq}
\def\ivee{\sqcup}
\def\imee{\sqcap}
\def\ibigvee{\bigsqcup}
\def\idirvee{\dirsqcup}

\DeclareMathOperator*{\dirsqcup}{\bigsqcup\!{}^{\uparrow}}
\newcommand\dirsue{\mathrel{\subseteq\!{}^{\uparrow}}}

\newcommand{\thickdot}{\scaleobj{1.3}{\cdot}}
\newcommand{\lleq}{\mathrel{\newcombinedsymbol{0.8}{\leq}{\ \,\thickdot}}}

\newcommand{\lvee}{\newsmalldotoperator{1.3}{\vee}}
\newcommand{\lmee}{\newsmalldotoperator{-1.8}{\wedge}}

\newcommand{\newsmalldotoperator}[2]{\mathbin{\newcombinedsymbol{#1}{#2}{\thickdot}}}

\newcommand{\newcombinedsymbol}[3]{\ensuremath{\ThisStyle{
    \ensurestackMath{\stackinset{c}{0\LMpt}{c}{#1\LMpt}{\SavedStyle#3}{\SavedStyle#2}}}}}

\newcommand\lpar{\mathbin{\rotatebox[origin=c]{180}{$\&$}}}
\newcommand\lolli{\multimap}

\newcommand\Chu{\mathtt{Chu}}

\begin{document}

\title{Four imprints of Belnap's useful four-valued logic in computer science}

\author*[1]{\fnm{Tom\'a\v s} \sur{Jakl}}\email{tomas.jakl@cvut.cz}

\affil*[1]{\orgdiv{Faculty of Information Technology}, \orgname{Czech Technical University},
    \orgaddress{
        
        \city{Prague},
        
        \country{Czech Republic}}
    }

\abstract{
    We review four areas of theoretical computer science which share technical or philosophical ideas with the work of Belnap on his useful four-valued logic. Perhaps surprisingly, the inspiration by Belnap--Dunn logic is acknowledged only in the study of d-frames. The connections of Belnap's work and linear logic, Blame Calculus or the study of LVars are not openly admitted.

    The key to three of these connections with Belnap's work go via the twist-product representation of bilattices.
    On the one hand, it allows us to view a large class of models of linear logic as based on Belnap--Dunn logic.
    On the other hand, d-frames admit two twist-product representation theorems and, also, the key theorem in Blame Calculus is essentially a twist-product representation theorem too, albeit with a strong proof-theoretic flavour.
}

\keywords{Belnap--Dunn logic, twist-construction, twist-product representation, d-frames, Blame Calculus, LVars}

\maketitle

\hfill{\emph{In memory of Nuel Belnap.}}

\section{Introduction}
\label{sec1}

Belnap's four-valued logic~\cite{belnap76,belnap77}, following Dunn~\cite{Dunn1976}, is an extension of the classical two-valued logic to encompass phenomena arising in computer-based reasoning systems. On top of the classical two values `True' and `False' it adds `None' representing the lack of any information and `Both' representing the inconsistency of claiming both `True' and `False'.
One of the key ideas in Belnap's approach is to distinguish between two orders of the four-valued lattice,
the usual `logical order'~\(\lleq\) and the `information/knowledge order'~\(\ileq\).
The latter order is inspired by domain theory, initiated by Scott~\cite{scott1970outline}. Domains were introduced to give semantics to (abstract) programming languages such as the lambda calculus.

The tradition of using ``Four Logical Alternatives,'' known as catu\d sko\d ti, goes back to the 6th century Indian logic.
These four alternatives precisely correspond to the aforementioned four values~\cite{robinson57indian,jayatilleke67indian,Garfield2019}.
However, undoubtedly, Belnap was the first to identify the two orders and to emphasise the usefulness of the four values for computer reasoning.

\begin{figure}[t]
    \begin{center}
\begin{tikzpicture}[node distance=1.5cm, font=\tiny]
    \coordinate (T) at (0,0);
    \coordinate (L) at (-1,-1);
    \coordinate (R) at ( 1,-1);
    \coordinate (B) at ( 0,-2);

    \path (T) edge (L) edge (R)
    (R) edge (B)
    (B) edge (L);

    \node at (0, 0.2) {Both};
    \node at (0,-2.2) {None};
    \node at (-1.5,-1) {False};
    \node at ( 1.5,-1) {True};

    \path[->,every node/.style={sloped,anchor=south,auto=false,font=\footnotesize, pos=0.6}]
        ($(L) + (-1.5,-1.7)$) edge node [below] {logical order $\lleq$} ($(R) + (+1, -1.7)$)
        ($(B) + (-2.0,-1.0)$) edge node {information order $\ileq$} ($(T) + (-2.0, 0.5)$);
\end{tikzpicture}
    \end{center}
    \caption{Belnap's four-valued lattice \FOUR.}
    \label{fig:four}
\end{figure}
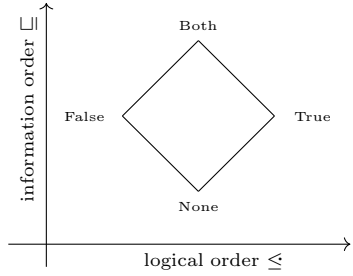

Later on Ginsberg~\cite{ginsberg1988multivalued} introduced bilattices as an algebraic description of the connectives introduced by Belnap.
As his primary examples he used world-based bilattices, that is, structures consisting of pairs \((A_+, A_-)\) of subsets of a chosen set of worlds.
Then, intuitively, \(A_+\) is a set of worlds where a given property holds and \(A_-\) are the worlds where it fails.
Under this interpretation four-valued connectives are calculated as simple set-theoretic operations. For example, the negation and logical conjunction are give by \(\neg(A_+, A_-) = (A_-, A_+)\) and \((A_+, A_-) \lmee (B_+, B_-) = (A_+ \cap B_+, A_- \cup A_-)\), respectively.

Then Fitting~\cite{fitting91bilattices} observed that the construction of Ginsberg holds more generally. That is, not only that the square products of a (powerset) lattice yields a bilattice but, in fact, this holds true for any lattice.
It is a remarkable result of Rivieccio, from his Ph.D.\ thesis~\cite{rivieccio10}, that every bilattice is obtained this way.
Philosophically speaking, this means that, even in this more general algebraic approach of Ginsberg~\cite{ginsberg1988multivalued}, every bilattice-valued proposition can be decomposed into two parts, its positive and negative parts.

In this paper we discuss three different examples from theoretical computer science where this tension between positive and negative parts is crucial.
Firstly, in Section~\ref{s:linear-logic}, we observe that the famous Chu construction of models of linear logic is exactly the twist-construction of bilattices, with most connectives being our well-known four-valued connectives of Belnap and Ginsberg.
Secondly, in Section~\ref{s:dfrm}, we discuss the author's own work on d-frames, which came from considerations in denotation semantics of programming languages and domain theory. On the one hand, d-frames provide a relaxation which encodes properties of topological models of Belnap--Dunn logic. Our approach reopens some of Belnap's original questions that have to do with computability in the information order.
After d-frames we turn to Blame Calculus (Section~\ref{s:blame-calc}). This is a type system for programming languages that allows to construct types from other types with set comprehension.

The attentive reader might notice that the title of this paper speaks of four imprints of Belnap's work, and not merely three. Indeed, the final section (Section~\ref{s:lvars}) is about a very practical application of ideas similar to Belnap's, in an implementation of special type of variables (called LVars) for programming purposes~\cite{KuperN13icfp,kuper2015thesis}. LVars were introduced to solve a particular problem in the theory distributed and parallel programs. Perhaps surprisingly, the core of the solution lies in adding the information order as part of the structure, which is very much in the spirit of Belnap who added the information order (and extra new values and connectives) to the usual two-valued logical lattice.

\section{Review of bilattices}

The algebraic properties of \FOUR{} (from Figure~\ref{fig:four}) are encapsulated by structures called \emph{bilattices}. These are the tuples \((B, \imee, \ivee, \lmee, \lvee, \neg, \mff, \mtt, \bot, \top)\) where
\begin{enumerate}[label=(BI-\expandafter\arabic*), labelindent=*, leftmargin=5em, labelsep=1em ]
    \item \((B, \imee, \ivee, \bot, \top)\) and \((B, \lmee, \lvee, \mff, \mtt)\) are (bounded) lattices\footnote{Throughout this text we always implicitly assume that our lattices are bounded.},
    \item \(x \ileq y\) implies \(\neg x \ileq \neg y\),
    \item \(x \lleq y\) implies \(\neg y \lleq \neg x\), and
    \item \(a = \neg \neg a\).
\end{enumerate}
We denote by \(\ileq\) the order induced by the lattice operations \(\imee,\ivee\) and, similarly, we denote \(\lleq\) the order induced by \(\lmee, \lvee\). Note that the constants \(\mff, \mtt, \bot, \top\) correspond to the four values `False', `True', `None', and `Both', respectively.

Furthermore, to not get lost in the intricate details concerning weak versions of bilattices, we always assume that our bilattices are \emph{distributive}, i.e.\ that any combination of pairs of operations from \(\imee, \ivee, \lmee, \lvee\) distributes over each other. In fact, the variety of (distributive) bilattices is generated by \FOUR{}, i.e.\ the bilattice on the set~\(\{\mff, \mtt, \bot, \top\}\), with \(\imee, \ivee, \lmee, \lvee\) as in Figure~\ref{fig:four} and with negation given by \(\neg \top = \top\), \(\neg \bot = \bot\) and \(\neg \mtt = \mff\) \cite[Theorem~2.4.1]{rivieccio10}.

\paragraph{Twist-construction}

In the following lines we describe a simple construction, that goes back to Kalman~\cite{kalman1958lattices} and which found its place not only in the theory of bilattices but also e.g.\ in the algebraic approach to relevance logic~\cite{tsinakis2006,fussner2019mingles}, albeit in a slightly different form.

Let~\(L_+\) and~\(L_-\) be distributive lattices, the \emph{twist-structure} \(L_+ \bowtie L_-\)is the algebra \((L_+\times L_-,\, \imee, \ivee, \lmee, \lvee, \mff, \mtt, \bot, \top)\) where
\begin{align*}
    (a_+, a_-) \imee (b_+, b_-) &= (a_+ \land b_+, a_- \land b_-), \\
    (a_+, a_-) \ivee (b_+, b_-) &= (a_+ \lor b_+, a_- \lor b_-), \\
    (a_+, a_-) \lmee (b_+, b_-) &= (a_+ \land b_+, a_- \lor b_-), \\
    (a_+, a_-) \lvee (b_+, b_-) &= (a_+ \lor b_+, a_- \land b_-), \\
    \mff = (0,1), \quad \mtt = (1,0),& \quad \bot = (0,0), \ete{and} \top = (1,1).
\end{align*}
The induced orders \(\ileq\) and \(\lleq\) are then given as follows.
\begin{align*}
    (a_+, a_-) \ileq (b_+, b_-) &\iff a_+ \leq b_+ \ete{and} a_- \leq b_- \\
    (a_+, a_-) \lleq (b_+, b_-) &\iff a_+ \leq b_+ \ete{and} b_- \leq a_-
\end{align*}

In case we use the same lattice twice, instead of writing \(L \bowtie L\) we simply write~\(L\twist\). Then, \(L\twist\) can be extended to a bilattice by interpreting \emph{negation} as
\[
    \neg(a_+, a_-) = (a_-, a_+).
\]

If, moreover, \(L\) is a Heyting algebra, we can view \(L\twist\) as an \emph{implicative} bilattice. This means that its signature is extended with \emph{weak implication} defined as follows.
\[
    (a_+, a_-) \oimply (b_+, b_-) = (a_+ \to b_+,\, a_+ \land b_-)
\]
Furthermore, the so-called \emph{strong implication} is now also term definable by
\begin{align}
    x \simply y = (x \oimply y) \lmee (\neg y \oimply \neg x)
    \label{eq:simply}
\end{align}
The reason for the distinction between these two implications is of practical nature. Weak implication supports the usual deduction theorem \cite{arieli96} whereas strong implication is residuated. That is, for the \emph{fusion} operation
\[
    x * y = \neg (y \simply \neg x)
\]
we have that \(x * y \lleq z\) iff \(x \lleq y \simply z\)~\cite[Proposition~5.4.1]{rivieccio10}.

It is easy to see that \FOUR{} is precisely \(\two\twist\), where \(\two\) is the two element Boolean algebra.
Furthermore, Rivieccio~\cite{rivieccio10} showed that that all `typical' bilattices are obtained this way, which extends the representation theorem of Avron~\cite{avron1996structure} for pre-bilattices.

\begin{lemma}
    \label{l:twist-bilat}
    For a bilattice \(B\), the following are equivalent.
    \begin{enumerate}
        \item \(B\) is \emph{interlaced}, that is:
            \begin{align*}
                x \lleq y &\ee{\implies} x \imee z \lleq y \imee z \ete{and} x \ivee z \lleq y \ivee z
                \\
                x \ileq y &\ee\implies x \lmee z \ileq y \lmee z \ete{and} x \lvee z \ileq y \lvee z
            \end{align*}
        \item \(B \cong L\twist\) for some bounded distributive lattice \(L\).
    \end{enumerate}
\end{lemma}

There are also similar representation theorems for implicative bilattices and many quasi-Nelson-type algebras~\cite{rivieccio2014implicative,RivieccioBusaniche2024conuclei}.

\section{Linear logic}
\label{s:linear-logic}

Girard's \cite{girard1987linear} linear logic enjoys a fair deal of popularity within the computer science community. This is due to its its ability to track resource usage.
In fact, linear logic was the main inspiration behind Rust~\cite{MatsakisKlock2014}, a relatively new and promising programming language whose strong points are safety guarantees of resource usage.

Linear logic is a substructural logic with connectives split into groups of multiplicatives \(\otimes, \lpar, \lolli\), additives \(\oplus, \&\), exponentials \(!,?\), and a linear negation \((\text{-})^\bot\).
It admits a very natural Brouwer--Heyting--Kolmogorov (BHK) interpretation, akin to that of intuitionistic logic. Shulman~\cite{Shulman2022affine} defines this in terms of the dynamics between \emph{proofs} (affirmations) and \emph{refutations}. For example, the interpretation of additive conjunctions turns out to be very familiar to anyone who has seen twist-constructions:
\begin{itemize}
    \item a proof of \(P \oplus Q\) is a proof of \(P\) and a proof of \(Q\), and
    \item a refutation of \(P \oplus Q\) consists of a refutation of \(P\) or a refutation of \(Q\).
\end{itemize}
Similarly, for negation we have that
\begin{itemize}
    \item a proof of \(P^\bot\) is a refutation of \(P\)
    \item a refutation of \(P^\bot\) is a proof of \(P\)
\end{itemize}
Perhaps the most unusual are the multiplicative connectives, e.g.\ the multiplicative disjunction \(\lpar\) (called `par') is interpreted as follows.
\begin{itemize}
    \item a proof of \(P \lpar Q\) is a method of converting a refutation of \(P\) into a proof of \(Q\) and method of converting a refutation of \(Q\) into a proof of \(P\), and
    \item a refutation of \(P \lpar Q\) is a refutation of \(P\) and a refutation of \(Q\).
\end{itemize}
The main restriction, when compared to paraconsistent logics, is that we do not allow statements to be both provable and refutable.

\paragraph{Chu construction}
The above BHK interpretation is justified by the so-called \emph{Chu construction}~\cite{chu1978constructing,chu1979constructing}, which is a standard method for building models of linear logic from models of intuitionistic logic. In the following we describe the \emph{proof-irrelevant} version of this construction following the \emph{proof-relevant} description of Shulman~\cite{Shulman2022affine}. This means that the same construction works not only for partially ordered structures but also for categories.

Given a Heyting algebra \(H\), the universe of \(\Chu(H,0)\) consists of pairs \(a = (a_+, a_-) \in H\times H\) such that \(a_+ \land a_- = 0\). In other words, we view propositions as pairs of intuitionistic statements (i.e.\ elements of \(H\)), representing the affirmation and refutation parts respectively. The latter condition forbids the statements to be both affirmable/provable and refutable.

The linear connectives are defined in \(\Chu(H,0)\) as follows:
\begin{align*}
    a \oplus b &= (a_+ \land b_+,\, a_- \lor b_-)
    &
    a \otimes b &= (a_+ \land b_+,\, (a_+ \to b_-)\land (b_+ \to a_-))
    \\
    a \mathbin{\&} b &= (a_+ \lor b_+,\, a_- \land b_-)
    &
    a \lpar b &= ((a_- \to b_+)\land (b_- \to a_+),\, a_- \land b_-)
    \\
    a^\bot &= (a_-,\, a_+)
    &
    a \lolli b &= ((a_+ \to b_+)\land (b_- \to a_-),\, a_+ \land b_-)
    \\
    !a &= (a_+,\, a_+ \to 0)
    &
    ?a &= (a_- \to 0,\, a_-)
\end{align*}
Note that \(!a\) is also sometimes defined as \((a_+, a_+ \to a_-)\). This is equivalent thanks to the assumption that \(a_+ \land a_- = 0\), since \(a_+ \to a_- \leq a_+ \to 0\) always holds and, by residuation, \(a_+ \land (a_+ \to a_-) = a_+ \land a_- \leq 0\) implies the converse \(a_+ \to 0 \leq a_+ \to a_-\). Dually, \(?a\) can be defined as \((a_- \to a_+, a_-)\).

It is also customary to equip \(\Chu(H,0)\) with constants \(\mtt,\mff\) defined as in \(H\twist\). Then,~\(a^\bot\) is definable from other connectives. In fact \(a^\bot\) is simply \(a \lolli \mff\), indeed:
\begin{align*}
    a \lolli \mff &= ((a_+ \to 0) \land (1 \to a_-),\, a_+ \land 1)
    \\
    &= ((a_+ \to 0) \land a_-, a_+)
    \\
    &= (a_-, a_+) = a^\bot
\end{align*}
where the penultimate equality holds from the fact that \(a_+ \land a_- = 0\) which yields \(a_- \leq a_+ \to 0\) by residuation.

Next, we observe that the structure of~\(\Chu(H,0)\) is term-definable from the bilattice structure of~\(H\twist\).

\begin{lemma}
    \label{l:ll-vs-bilat}
    The following holds for any \(a,b\in \Chu(H,0)\):
    \begin{align*}
        a \oplus b &= a \lmee b
        &
        a \otimes b &= \neg (a \simply \neg b) = a * b
        \\
        a \mathbin{\&} b &= a \lvee b
        &
        a \lpar b &= \neg a \simply b
        \\
        a^\bot &= \neg a
        &
        a \lolli b &= a \simply b
        \\
        !a &= \neg (a \oimply \mff)
        &
        ?a &= (\neg a) \oimply \mff
    \end{align*}
\end{lemma}
\begin{proof}
    The only expressions that need justification are those for \(a \otimes b\), \(a \lpar b\) and \(a \lolli b\).
    For the latter, by definition we have:
    \begin{align*}
        a \simply b
        &= (a \oimply b) \lmee (\neg b \oimply \neg a)
        \\
        &= (a_+ \to b_+,\, a_+ \land b_-) \lmee (b_- \to a_-,\, b_- \land a_+)
        \\
        &= ((a_+ \to b_+)\land (b_- \to a_-),\, (a_+ \land b_-) \lor (b_- \land a_+))
        = a \lolli b
    \end{align*}

    Next, observe that \(\neg(a \otimes b) = ((a_+ \to b_-)\land (b_+ \to a_-),\, a_+\land b_+) = a \simply \neg b\) from which we obtain the desired formula for \(a \otimes b\).
    Lastly, the formula for \(a \lpar b\) follows from the fact that \(a \lpar b = \neg (\neg a \otimes \neg b)\) and \(\neg \neg a = a\).
\end{proof}

One of the original motivation behind linear logic in Girard~\cite{girard1987linear} was to better understand implication, by expressing it in terms of more fine-grained connectives. This corresponds to the embedding \(H \to \Chu(H,0)\), \(a \mapsto (a, a\to 0)\), where intuitionistic implication \(a \to b\) gets interpreted as \(!a \lolli b\). It turns out that this is precisely weak implication in bilattice logic:
\begin{align*}
    !a \lolli b
    &= (a_+, a_+ \to 0) \lolli b
    \\
    &= ((a_+ \to b_+)\land (b_- \to (a_+ \to 0)),\, a_+ \land b_-)
    \\
    &= ((a_+ \to b_+)\land (a_+ \to (b_- \to 0)),\, a_+ \land b_-)
    \\
    &= ((a_+ \to b_+),\, a_+ \land b_-) = a \oimply b
\end{align*}
where the penultimate equality holds because \(a_+ \to b_+\leq a_+ \to (b_- \to 0)\). Indeed, intuitionistic implication~\(\to\) is monotone in its second argument and we have that \(b_+ \leq b_- \to 0\) from \(b_+ \land b_- = 0\) by residuation.

Note that models of linear logic of the form \(\Chu(H,0)\) satisfy some extra theorems that are not typically true in linear logic, such as \(!!P \equiv !P\), \(!P \equiv P \otimes P\), \(P \oplus (Q \& R) \equiv (P\oplus Q) \& (P\oplus R)\), etc.
As Shulman~\cite{Shulman2022affine} points out, this makes the logic of \(\Chu(H,0)\) uninteresting to linear logicians but, on the other hand, it provides a useful tool for studying constructive mathematics, by being a richer logic that allows for more precise distinctions in constructive statements when compared to intuitionistic logic. Given our observation above that the linear connectives of the Chu construction are term definable in Belnap--Dunn logic we might suggest that, in fact, Belnap--Dunn logic might be even better suited for studying constructive statements.

\begin{remark}
    There is another way to obtain twist-product models of linear logic.
    Indeed, Avron~\cite{avron1988ll} makes a link between linear logic and relevance logic where, furthermore, the latter was shown to leverage twist-products for model constructions~\cite{tsinakis2006,fussner2019mingles}.
    What is interesting about Lemma~\ref{l:ll-vs-bilat} is that it gives a direct and simple description of linear logic connectives in terms of bilattice connectives which, moreover, is an exact correspondence in most cases.
\end{remark}

\section{d-Frames}
\label{s:dfrm}

The second instance where four-valued logic occurs in theoretical computer science has the same origins as Belnap--Dunn logic, that is, in domain theory.

\paragraph{Continuous domains}
\emph{Domains} are partially ordered structures \((D, \ileq)\) with the extra requirement that its \emph{directed subsets} \(A \dirsue D\) have suprema, denoted by \(\idirvee A\).
Intuitively, domains encode the information obtained during some computation. The order \(a \ileq b\) encodes that, at state~\(a\), the computed information is less than or equal to that of \(b\). Consequently, directed suprema~\(\idirvee\) model computation, e.g.\ by means of recursion, loops and alike.

As argued already by Scott~\cite{scott1970outline}, in order for a domain to correspond to a real computation, it has to be of a special type.
It has to have a countable (computer enumerable) base \(B \sue D\), which we think of as `finite pieces' of information from which the rest of the domain is computed.

To make this precise, for \(k,l \in D\), we say that \emph{\(k\) approximates \(l\)}, and write \(k \ll l\), if for any \(A \dirsue D\)
\[
    l \ileq \idirvee A \ee\implies k \ileq a \text{ for some \(a\in A\)}.
\]
Then, the domain \((D,\ileq)\) is \emph{continuous} if it has a base \(B\) such that, for every \(x\in D\), \(x = \idirvee \{ b\in B \mid b \ll x\}\).
This models the idea that every state in our domain can be computed from finite pieces of information that approximate it.
Note that the requirement that the base is countable is often dropped.

An important feature of domains is that they can be viewed as \emph{topological spaces}. A domain~\(D\) can be equipped with the topology of Scott-open sets, i.e.\ upwards closed sets \(U\sue D\) such that \(\idirvee A \in U\) implies \(A\cap U \not= \emptyset\).
This opens up the door for connections between computation, domain theory and (non-Hausdorff) topology, which has been a fruitful field of study over the years, see e.g.~\cite{mislove1998topology,escardo2004synthetic,GoubaultLarrecq2013}.

\paragraph{Bitopological spaces and their duals}
Domains were introduced by Scott~\cite{scott1970outline}
as an alternative to the untyped lambda calculus for uses in semantics. Famously, by coincidence, domains also played a key role in providing a concrete semantics to the untyped lambda calculus~\cite{scott1993}, which was an important open problem at that time.
Since then they were used to model other programming phenomena such as non-determinism, exceptions, and so on.

However, it is still not known how to model probabilistic computation with continuous domains\footnote{Formally, it is not known if there is a cartesian closed subcategory of continuous domains which is closed under the probabilistic powerdomain construction.}, which remains a major challenge in the field.

The search for the solution accelerated the study of topological properties of domains (via their Scott topology) and, in particular, of the class of \emph{stably compact spaces} which correspond to a well-behaved class of continuous distributive lattices, see e.g.~\cite{lawson2011stably,GoubaultLarrecq2013}.

A distinct feature of stably compact spaces is that their open and compact subsets are dual to each other, in some sense.
In fact, compact subsets generate the so called \emph{dual topology}, which plays an important role the theory.
This naturally led Jung and Moshier~\cite{jung06} to focus on bitopological spaces, as a convenient tool for describing properties of stably compact spaces in terms of their dual topology.
To this end, they introduced a new type of algebraic duals of bitopological spaces, called \emph{d-frames}, and this is where our story turns back to Belnap.

Recall that any topological space \((X,\tau)\) induces the lattice \(L = (\tau, \sue)\) of its open sets in the subset order. This lattice is a complete Heyting algebra or, equivalently, it is a complete lattice such that, for any \(A\sue L\) and \(b\in L\), the equation
\[(\bigvee A)\wedge b = \bigvee\nolimits_{a\in A} (a\wedge b)\]
holds. Complete lattices satisfying this distributivity law are also often called \emph{frames} or \emph{locales}~\cite{picadopultr2011}.
For a large class of topological spaces, i.e.\ for the so-called \emph{sober spaces} (which includes all Hausdorff spaces), the frame of opens \(L\) contains all information about the space \(X\) and, in fact, the set of points of \(X\) can be fully recovered from \(L\).

In case we have a \emph{bitopological space} \((X,\tau_+,\tau_-)\), which constitutes of a set \(X\) and two topologies \(\tau_+,\tau_-\), we have two frames \(L_+ = (\tau_+, \sue)\) and \(L_- = (\tau_-, \sue)\).
In order to describe the interaction of the two topologies, Jung and Moshier \cite{jung06} add two relations \(\con,\tot \sue L_+\times L_-\) which represent pairs of disjoint subsets and pairs of subsets that cover the whole space, respectively.
This is all is bundled into one structure as follows.

\begin{definition}
    A \emph{d-frame} is a tuple \((L_+, L_-, \con, \tot)\) where \(L_+, L_-\) are frames and \(\con, \tot \sue L_+\times L_-\) are binary relations satisfying the following axioms:
    \begin{enumerate}
        \item \(\con\) and \(\tot\) are logical sublattices of \(L_+ \bowtie L_-\), i.e.\ they contain \(\mff = (0,1)\), \(\mtt = (1,0)\) and are closed under~\(\lmee, \lvee\),
        \item \(\con\) is downwards closed and \(\tot\) is upwards closed in the \(\ileq\) order of~\(L_+ \bowtie L_-\),
        \item \(\con\) is closed under directed joins \(\idirvee\) in~\(L_+ \bowtie L_-\),
        \item \(\con(a,b)\) and \(\tot(a,c)\) implies \(b\leq c\) and, dually,\\
              \(\con(a,b)\) and \(\tot(c,b)\) implies \(a\leq c\).
    \end{enumerate}
\end{definition}

It is immediate that any bitopological space \((X,\tau_+, \tau_-)\) yields the d-frame \((\tau_+, \tau_-, \con_X, \tot_X)\) where \(\con_X(U_+, U_-)\) whenever \(U_+\cap U_- = \emptyset\) and \(\tot_X(U_+,U_-)\) whenever \(U_+\cup U_- = X\).
Item 4 in the axiomatisation of d-frames is inspired by this construction of d-frames from bitopological spaces. It is simply an instance of the set-theoretic fact that \(U \cap V = \emptyset\) and \(W \cup V = X\) imply \(U \subseteq W\).

The notation for \(\con\) refers to `consistency' from logic.
For example, similarly to Ginsberg~\cite{ginsberg1988multivalued}, assume \(X\) is a set of words or models of some theory.
Then, a pair of opens \((U_+, U_-)\in \tau_+ \times \tau_-\) represents the set of models \(U_+\) where a given predicate holds and the set of models \(U_-\) where it fails.
Then, \(\con(U_+,U_-)\) represents that there is no model where the predicate would be both true and false, i.e.\ the predicate is consistent.
On the other hand, a predicate is `total', formally written as \(\tot(U_+,U_-)\), if for every model~\(x\in X\), either \(x\) satisfies the predicate (i.e.\ \(x\in U_+\)) or \(x\) fails it (i.e.\ \(x\in U_-\)).

A special type of predicates are those for which both \(\con_X(U_+,U_-)\) and \(\tot_X(U_+, U_-)\). These can be referred to as \emph{classical} since we have that \(U_+\) and \(U_-\) are complements of each other.

\begin{remark}[Twist representation]
    As is evident from the axiomatisation, despite keeping the two lattices \(L_+, L_-\) separate, we view d-frames as twist-structures \mbox{\(L_+ \bowtie L_-\)} with two unary predicates \(\con, \tot \sue L_+ \bowtie L_-\).
    Conversely, there is a twist-structure decomposition theorem similar to Lemma~\ref{l:twist-bilat} for d-frames. They can be viewed as structures \((L, \mff, \mtt, \con, \tot)\) where \((L,\ileq)\) is a frame, \(\mff,\mtt\) are two complemented constants in \(L\) and \(\con, \tot\) are two unary relations axiomatised similarly to the above. The logic order \(\lleq\), necessary for the axiomatisation, is recovered from~\(\mff,\mtt\) and~\(\ileq\) as follows~\cite[page 27]{jung06}.
\[
    \alpha \lleq \beta \iff (\alpha \imee \mtt) \ileq (\beta \imee \mtt)
    \ete{and} (\beta \imee \mff) \ileq (\alpha \imee \mff)
\]
\end{remark}

\begin{remark}[On completeness, continuity and computation]
    \label{r:compl+cont-dfrm}
    We might be led to think of d-frames as of bilattices with additional relations.
    However, we also have the extra assumption that the individual components \(L_+\) and \(L_-\) are complete lattices or, in fact, complete Heyting algebras.

    Note that completeness in the information order \(\ileq\) was assumed by Belnap from the very beginning, cf.~\cite{belnap77}.
    In fact, Belnap also spoke of computability aspects and by that he meant that the structure \((L_+ \bowtie L_-, \ileq)\) is a continuous domain, in the sense that we discussed above. This is evident either from the last paragraphs of \S 81.1 or from \S 81.3.2 of~\cite{belnap77}.

    Notably, both of these conditions were left out in the work of Belnap's followers.
    In the setting of d-frames, on the other hand, the computability aspects reappear automatically.
    As discussed earlier, stably compact spaces \((X,\tau)\) correspond to bitopological spaces \((X,\tau_+, \tau_-)\) where \mbox{\(\tau_+ = \tau\)} and \(\tau_-\) is the dual topology to \(\tau\). Then, the corresponding d-frame \((\tau_+, \tau_-, \con_X, \tot_X)\) satisfies that \((\tau_+ \bowtie \tau_-, \ileq)\) is a continuous domain~\cite{jung06,jakl2018}.
    Finally, in Section~\ref{s:partial-frm} below, we look at the `partial frame' of consistent elements of a d-frame. It turns out that every partial frame arising from a stably compact space is a continuous domain too.
\end{remark}

\subsection{Bilattices \emph{versus} d-frames}

Since d-frames are based on the twist-structure \(L_+ \bowtie L_-\), they interpret a lot of the logic of bilattices. However, importantly, they lack negation and weak implication. This is because they are \emph{non-symmetric}, i.e.\ the lattices \(L_+\) and \(L_-\) are potentially different, and so we cannot define \(\neg\) and \(\oimply\) in the same way as for bilattices.
On the other hand, the language of d-frames contains the consistency \(\con\) and totality \(\tot\) relations and directed suprema in the knowledge order \(\idirvee\).
These differences are summarised in the first two columns in Figure~\ref{l:bilat-vs-dfrm}. 

\begin{figure}[ht]
    \begin{center}
    \begin{tabular}{c | c | c}
        bilattices & d-frames & nd-frames \\
        \hline
        carrier symmetric & carrier non-symmetric & carrier non-symmetric\\
        ($L\times L$)    & ($L_+\times L_-$)    & ($L_+\times L_-$)      \\[0.5em]
        $\neg$, $\oimply$ & ---                   & $\neg$, $\oimply$    \\[0.5em]
        ---               & $\con$, $\tot$, $\idirvee$  & $\con$, $\tot$, $\idirvee$ \\[0.5em]
    \end{tabular}
    \end{center}
    \label{l:bilat-vs-dfrm}
    \caption{Comparison of bilattice and d-frame structure.}
\end{figure}

Note that, for a given lattice \(L\), the bilattice \(L\twist\) can be extended with canonical \(\con\twist\) and \(\tot\twist\) as follows:
\[
    \con\twist(a,b) \iff a\land b = 0
    \qtq{and}
    \tot\twist(a,b) \iff a\lor b = 1
\]
Furthermore, if \(L\) is complete then \(L\twist\) also has directed suprema \(\idirvee\), computed componentwise in \(L\)\footnote{In fact, if \(L\) is a frame, it follows that \((L, L, \con\twist,\tot\twist)\) is a d-frame.}.
On the other hand, there are other possible consistency and totality relations on \(L\twist\) or, in fact, on any \(L_+ \bowtie L_-\). For example, we have the following trivial relations.
\[
    \con^\text{triv}(a,b) \iff a=0 \ete{or} b=0
    \qtq{and}
    \tot^\text{triv}(a,b) \iff a=1 \ete{or} b=1
\]
The axiomatisation of d-frames allows for these and many more types of consistency and totality relations.

However, probably the main distinguishing factors of d-frames from bilattices is that they are non-symmetric and have a natural topological interpretation.
In fact, in the next paragraphs we discuss how we can use the intuition from bitopological spaces to extend the signature of d-frames, so that they interpret also \(\neg\) and \(\oimply\), to subsume both bilattices and d-frames, as shown in the third column of Figure~\ref{l:bilat-vs-dfrm}.

\paragraph{Interior operators}

Looking at the formulas for \(\neg\) and \(\oimply\) for bilattices from their twist-structure, in order to define similar operations for d-frames, we need some way of interpreting each element \(a_+\in L_+\) as an element of \(L_-\) and vice versa.
Following \cite{jakljungpultr2016,jakl2018} we solve this puzzle by looking for inspiration in bitopological spaces. Given a bitopological space \((X,\tau_+,\tau_-)\) and an open set \(U_+\in \tau_+\), we can interpret it as the open set with respect to the other topology \(\tau_-\) by taking its \(\tau_-\)-interior. Conversely, any \(U_-\in \tau_-\) can be sent to \(\tau_+\) by taking the \(\tau_+\)-interior of \(U_-\).
Formally, we define the following two maps.
\begin{align}
    m &\colon \tau_+ \to \tau_-,& U_+ \mapsto \bigcup \{ V_-\in \tau_- \mid V_- \sue U_+\}
    \label{eq:m}
    \\
    p &\colon \tau_- \to \tau_+,& U_- \mapsto \bigcup \{ V_+\in \tau_+ \mid V_+ \sue U_-\}
    \notag
\end{align}

Inspired by this, we extend the signature of d-frames with extra two maps which are required to satisfy properties that we know hold for interior operators.

\begin{definition}
    An \emph{nd-frame} \((L_+, L_-, \con, \tot, p,m)\) is a d-frame equipped with maps \(p\colon L_- \to L_+\) and \(m\colon L_+ \to L_-\) satisfying the following properties:
    \begin{enumerate}
        \item \(p(a \land b) = p(a) \land p(b)\) and \(m(a\land b) = m(a) \land m(b)\),
        \item \(p(1) = 1\) and \(m(1) = 1\),
        \item \(p(0) = 0\) and \(m(0) = 0\),
        \item \(p(m(a)) \leq a\) and \(m(p(a)) \leq a\),
        \item \(\con\) axioms:
        $$
        \frac{(a\wedge b,c)\in \con}{(a,m(b)\wedge c)\in \con}
        \qquad
        \frac{(a,b\wedge c)\in \con}{(a\wedge p(b),c)\in \con}
        $$

        \item \(\tot\) axioms:
        $$
        \frac{(a,m(b)\vee c)\in \tot}{(a\vee b,c)\in \tot}
        \qquad
        \frac{(a\vee p(b),c)\in \tot}{(a,b\vee c)\in \tot}
        $$
    \end{enumerate}
\end{definition}

As expected, given any bitopological space \((X,\tau_+, \tau_-)\), its induced d-frame \((\tau_+, \tau_-, \con_X, \tot_X)\) equipped with \(p\) and \(m\), as defined in \eqref{eq:m}, is an nd-frame.

Then, with the structure of nd-frames, we can define the missing logical connective by adapting the corresponding formulas for bilattices. Given \(x, y\in L_+\times L_- \), define
\begin{align}
    \neg x = (p(x_-), m(x_+))
    \qtq{and}
    x \oimply y = (x_+ \to y_+, m(x_+) \land y_-).
    \label{eq:neg-opimply-ndfrm}
\end{align}

We can now compare the logical structure of nd-frames with that of bilattices. The following shows that most of the axioms of the Hilbert-style axiomatisation of Belnap--Dunn logic, given by Arieli and Avron~\cite{arieli96}, still holds for nd-frames.
We write \(\alpha\equiv \beta\) as a shorthand for \((\alpha \oimply \beta) \lmee (\beta \oimply \alpha)\). Also, we say that a formula \(\varphi\) \emph{is valid} in a d-frame~\(\p L\) if for any valuation \(v\) of variables in \(\p L\), \(v(\varphi) = v(\varphi \oimply \varphi)\) i.e.\ if \(\mtt \ileq v(\varphi)\) in \(\p L\).

\begin{theorem}[Theorem~4.2 in \cite{jakljungpultr2016}]
        The following axioms of four-valued logic are valid in any nd-frame:\\[0.3em]
(Weak implication)
\begin{axioms}
\item[($\oimply 1$)]  $\varphi \oimply (\psi\oimply \varphi)$
\item[($\oimply 2$)]  $(\varphi \oimply (\psi\oimply \gamma)) \oimply ((\varphi\oimply \psi) \oimply (\varphi\oimply \gamma))$
\item[($\neg\neg$ R)] $\neg \neg \varphi \oimply \varphi$ \addStar{A}
\end{axioms}
(Logical conjunction and disjunction)
\begin{axioms}
\item[($\wedge \oimply$)] $(\varphi\wedge \psi)\oimply \varphi$ and $(\varphi\wedge \psi)\oimply \psi$
\item[($\oimply \wedge$)] $\varphi\oimply (\psi \oimply (\varphi\wedge \psi))$
\item[($\oimply \dtt$)]   $\varphi\oimply \dtt$
\item[($\oimply \vee$)]   $\varphi\oimply (\varphi\vee \psi)$ and $\psi\oimply (\varphi\vee \psi)$
\item[($\vee \oimply$)]   $(\varphi \oimply \gamma)\oimply ((\psi\oimply \gamma)\oimply ((\varphi\vee \psi)\oimply \gamma))$
\item[($\oimply \dff$)]   $\dff \oimply \varphi$
\end{axioms}
(Informational conjunction and disjunction)
\begin{axioms}
\item[($\sqcap \oimply$)] $(\varphi\sqcap \psi)\oimply \varphi$ and $(\varphi\sqcap \psi)\oimply \psi$
\item[($\oimply \sqcap$)] $\varphi\oimply (\psi \oimply (\varphi\sqcap \psi))$
\item[($\oimply \top$)]   $\varphi\oimply \top$
\item[($\oimply \sqcup$)] $\varphi\oimply (\varphi\sqcup \psi)$ and $\psi\oimply (\varphi\sqcup \psi)$
\item[($\sqcup \oimply$)] $(\varphi \oimply \gamma)\oimply ((\psi\oimply \gamma)\oimply ((\varphi\sqcup \psi)\oimply \gamma))$
\item[($\oimply \bot$)]   $\bot \oimply \varphi$
\end{axioms}
\noindent(Negation)
\begin{axioms}
\item[($\neg\wedge$ L)]      $\neg(\varphi\wedge \psi) \oimplied \neg \varphi\vee   \neg \psi$ \addStar{B}
\item[($\neg~\vee$)]         $\neg(\varphi\vee   \psi) \equiv    \neg \varphi\wedge \neg \psi$
\item[($\neg~\sqcap$)]       $\neg(\varphi\sqcap \psi) \equiv    \neg \varphi\sqcap \neg \psi$
\item[($\neg\sqcup$ L)]      $\neg(\varphi\sqcup \psi) \oimplied \neg \varphi\sqcup \neg \psi$ \addStar{B}
\item[($\neg\!\!\oimply$ R)] $\neg (\varphi\oimply \psi) \oimply \varphi\wedge \neg \psi$ \addStar{A}
\end{axioms}

\noindent
Furthermore, the rule of Modus Ponens is sound:
\begin{axioms}
\item[(MP)] $\varphi, (\varphi \oimply \psi) \vdash \psi$
\end{axioms}
\end{theorem}

The axioms denoted by \theStar{A} or \theStar{B} are the only axioms that differ from the axioms of \cite{arieli96}, by only being implications in one direction instead of equivalences.
Requiring equivalence instead of implication in the axioms marked by \theStar{A} is the same as requiring that $p\circ m = \id$ and
the same for axioms marked by \theStar{B} corresponds to requiring that $p$ preserves finite suprema.
Also, in some presentations of bilattices, e.g.\ in \cite{arieli96}, the following axiom called \emph{Peirce's law} is added
\begin{axioms}
\item[($\oimply 3$)] $ ((\varphi\oimply \psi) \oimply \varphi)\oimply \varphi $
\end{axioms}
Assuming this to hold is equivalent to assuming that $L_+$ is a Boolean algebra.

\begin{remark}
We see that d-frames give rise to models of Belnap--Dunn logic. The algebraic operations \(\neg\) and \(\oimply\) were derived semantically, from the topological models. However, the properties of \(\neg\) and \(\oimply\) also have a good interpretation purely logically, which comes from a logical interpretation of the maps \(p\) and \(m\).

We can view \(p\) and \(m\) as maps that switch context, e.g.\ \(p(a)\) interprets the negative information \(a\) in the positive context.
Such change of context can lead to information loss, which is witnessed by axiom (4) of nd-frames and, correspondingly, the fact that \(\neg \neg x \ileq x\) and that ($\neg\neg$ R) holds.

    The same approach of adding negation to d-frames can be adapted to other non-symmetric situations. For example, Rivieccio, Maia and Jung~\cite{rivieccio2020non} equip bilattices \(L_+ \bowtie L_-\) with extra maps \(p\colon L_- \to L_+\) and \(m\colon L_+ \to L_-\), which satisfy the axioms (1)--(4) of nd-frames, to obtain non-symmetric bilattices with negation and implications.
\end{remark}

\begin{remark}[Belnap--Dunn geometric logic]
Note that the usual logic of frames is the so-called \emph{geometric} logic, which is an infinitary logic with arbitrary disjunctions \(\bigvee\) and finitary conjunctions \(\land\). This logic has its d-frame counterpart, described in detail in Section~6.3 of \cite{jakl2018}. It is built out of arbitrary disjunctions and finite conjunctions \(\ibigvee, \imee\) in the information order, and \(\con(\alpha)\) and \(\tot(\alpha)\) predicates, corresponding to the two relations of d-frames.

\end{remark}

\paragraph{Models of four-valued logics from continuous maps}
Let us now look again at the situation where an nd-frame \((L_+, L_-, \con, \tot, p,m)\) validates the stronger version of axioms \theStar{A}, in which the weak implication is replaced by equivalence \(\equiv\). We remarked that this is equivalent to $p\circ m = \id$. In this case we have that, for any \(a\in L_-\) and \(b\in L_+\),
\[
    m(p(b)) \leq b
    \qtq{and}
    a \leq p(m(a))
\]
where the former is from axiom 4 of nd-frames. In other words, we have that \(m\) is the left adjoint to \(p\) and, therefore, \(m\) preserves all suprema. Furthermore, by axioms~1 and~2 of nd-frames, \(m\) preserves finite infima. We obtain that \(m\) is a \emph{frame homomorphism}. In fact, since $p\circ m = \id$, \(m\) is an embedding of frames \(L_+ \hookrightarrow L_-\).

Conversely, given any frame embedding \(m\colon L \hookrightarrow M\), we define the structure \((L,M, \con_m, \tot_m, m, p_m)\) as follows
\[
    \con_m(a,b) \iff m(a) \land b = 0
    \qtq{and}
    \tot_m(a,b) \iff m(a) \lor b = 1
\]
and \(p_m\) is simply the localic map associated to \(m\), i.e.\ it is the right adjoint to \(m\).
A~routine verification shows that this structure is an nd-frame such that \(p \circ m = \id\).
In fact, these two constructions are inverse to each other.

\begin{proposition}
    \label{p:inj-frm-homo}
    There is a one-to-one correspondence between injective frame homomorphisms and nd-frames satisfying \(\varphi \oimply \neg\neg \varphi\).
\end{proposition}
\begin{proof}
    
    What is left to show is that for any nd-frame \((L_+, L_-, \con, \tot, p,m)\) satisfying \(\varphi \oimply \neg\neg \varphi\), i.e.\ such that \(p \circ m = \id\), we have that \(\con = \con_m\), \(\tot = \tot_m\) and \(p = p_m\).

    The equality \(p = p_m\) holds from uniqueness of adjoints.
    For \(\tot_m \sue \tot\), if \mbox{\(m(a) \lor b = 1\)} then \((0, m(a) \lor b)\) is in \(\tot\) since it is equal to \(\mff = (0,1)\). Consequently, by axiom 6 of nd-frames, \((a, b) = (0 \lor a, b)\) is in \(\tot\).
    Conversely, assume \(\tot(a,b)\). Since \(p\circ m = \id\), we have \(\tot(pm(a),b)\) and so, again by axiom 6 of nd-frames, \(\tot(0, m(a)\vee b)\). However, since \(\mff\) is in \(\con\), we obtain \(1 \leq m(a) \vee b\) by axiom 4 of d-frames, i.e.\ \(\tot_m(a,b)\).

    Finally, for \(\con \sue \con_m\), let \(\con(a,b)\). By axiom 5 of nd-frames, \(\con(1, m(a)\land b)\) which implies \(m(a) \land b \leq 0\) by axiom 4 of d-frames since \(\mtt\) is in \(\tot\).
    Conversely, given \(\con_m(a,b)\) i.e.\ \(m(a)\land b = 0\), we have that \(\con(1,m(a) \land b)\) which by axiom 5 of nd-frames implies \(\con(1 \land pm(a), b)\). However, \(p\circ m = \id\) entails that also \(\con(a,b)\).
\end{proof}

This gives us a fairly flexible procedure for obtaining models of (a version of) Belnap--Dunn logic.
Any continuous map between topological spaces \(f\colon X\to Y\) yields a frame homomorphism \(h\colon L \to M\) where \(L\) and \(M\) are the frame of opens of \(Y\) and \(X\) respectively. The map \(h\) sends the open \(U \sue Y\) to the preimage \(f^{-1}[U]\sue X\), which is open too.
In order to obtain an injective frame homomorphism from \(h\) it is enough factor it into a surjective homomorphism \(L \twoheadrightarrow L'\) followed followed by an injective homomorphism \(m\colon L' \hookrightarrow M\)\footnote{In fact, this is precisely the frame homomorphism obtained from the restriction of \(f\) to its image \mbox{\(X \twoheadrightarrow f[X]\)}.}
This then, by Proposition~\ref{p:inj-frm-homo}, yields an nd-frame \((L', M, \con_m, \tot_m, p_m, m)\).

From the point of view of the computational interpretation of \cite{escardo2004synthetic}, \(M\) in this nd-frame is the lattice of decidable properties of \(Y\) (i.e.\ the opens of \(Y\)), \(L\) consists of the decidable properties of \(X\) and \(p\) and \(m\) are computed in terms of the continuous map~\(f\), i.e.\ by a computable function which translates from the domain of \(X\) to that of~\(Y\).

\subsection{Partial frames and comparison to linear logic}
\label{s:partial-frm}

Jung and Moshier~\cite{jung06} observed that d-frames admit another twist-type representation theorem. They observed that, given a d-frame \((L_+, L_-, \con, \tot)\), the domain of consistent elements \((\con, \ileq)\) in the information order, when equipped with an extra relation
\[
    x \prec y \iff \tot(y_+, x_-),
\]
fully represents the original d-frame. Formally, we have the following.

\begin{proposition}[Theorem 7.5 of~\cite{jung06}]
    \label{p:partial-frm}
    The consistent elements of a d-frame, viewed as a substructure in the signature \(\left<\prec,\, \idirvee, \imee, \bot,\, \lmee, \lvee, \mff, \mtt\right>\), is a \emph{`partial frame'}\footnote{We do not give the definition of partial frames as it would drift us away from the main story. However, it is not complicated, see Section~7 in \cite{jung06}.}.

    Conversely, given a `partial frame' \(P\), the structure \(([\bot, \mtt], [\bot, \mff], \con_P, \tot_P)\) where \([a,b] = \{ x \in P \mid a \ileq x \ileq b]\) and
    \[
        \con(a,b) \iff \exists c.\, a \ileq c \ete{and} b \ileq c
        \qquad
        \tot(a,b) \iff b \prec a
    \]
    is a d-frame. Up to isomorphism, this is a one-to-one correspondence.
\end{proposition}

If we start from an nd-frame \((L_+, L_-, \con, \tot, p,m)\) instead of just a d-frame, we can compare the relation \(\prec\) and the weak~\(\oimply\) and strong~\(\simply\) implications, as defined in \eqref{eq:neg-opimply-ndfrm} and \eqref{eq:simply}, respectively.

\begin{lemma}[Theorem 4.4~\cite{jakljungpultr2016}]
    For \(x,y\in \con\), if \(x \prec y\) then both \(x \oimply y\) and \(x \simply y\) are valid.
\end{lemma}

Furthermore, negation \(\neg\colon L_+ \bowtie L_- \to L_+ \bowtie L_-\), defined from \(p,m\) as in \eqref{eq:neg-opimply-ndfrm}, restricts by axiom 5 of nd-frames to \(\neg\colon \con \to \con\), i.e.\ \(\neg\) is an operation on the partial frame of consistent elements.
Axiom 6 of nd-frames immediately yields how~\(\prec\) interacts with negation:
\[
    \frac{x\prec\neg y\vee z}{x\wedge y\prec z}
    \qtq{and}
    \frac{x\wedge\neg y\prec z}{x\prec y\vee z}
\]
Furthermore, we have the following cut rule~\cite[Section 4.1]{jakljungpultr2016}:
\[
    \frac{x\prec \neg\neg y\vee z \qquad z\wedge \neg\neg x'\prec  y'}{ x\wedge x' \prec y\vee  y'}
\]

Observe now that, for a complete Heyting algebra \(H\), i.e.\ a frame, the Chu construction \(\Chu(H,0)\) from Section~\ref{s:linear-logic} is based on the same elements as the partial frame corresponding to the twist-product d-frame \(H\twist\).
Indeed, the partial frame is based on pairs \((a,b)\in H\times H\) such that \(\con\twist(a,b)\) which, in turn, is equivalent to \(a\land b = 0\).

Consequently, Proposition~\ref{p:partial-frm} tells us that linear logic connectives can be represented in terms of the partial frame structure\footnote{Note that in order to interpret \((-)^\bot\), we need the twist-product d-frame \(H\twist\) to be an nd-frame, where \(p\) and \(m\) are just the identity functions. Then, \((-)^\bot\) is just \(\neg\) restricted to the partial frame corresponding to~\(H\twist\).}.
It is perhaps a bit surprising that the information order part of the structure of a partial frame, i.e.\ the connectives \(\idirvee, \imee, \bot\), do not seem to show up in the study of Chu constructions.
Also, the extra relation \(\prec\) of partial frames represents a form of strong implication which, however, seems absent not only in the linear logic literature but also in the bilattice literature.

On the other hand, linear logic of Shulman~\cite{Shulman2022affine}, makes use of the following order relation in order to define subobject classifiers. See the early draft~\cite[Theorem 4.7]{Shulman2018linear} where it appears more explicitly.
\[
    x \lhd y \iff x_+ \leq y_+ \ete{and} y_- \land (y_+ \to 0) \leq x_-
\]
Exploring the interplay and relationships of these notions, whether in the context of bilattices, d-frames or Chu constructions might be an interesting source of new results.

\section{Blame Calculus}
\label{s:blame-calc}

We now leave the world of algebraic structures underpinning various logics and take a look at another example in theoretical computer science which shares some similarities with Belnap's work.

To ensure (partial) correctness of computer programs many programming languages use the mechanism of `typing' which allows the programmer to give a specification of output and input values of their programs. This specification is then checked by the compiler and programs that do not `typecheck' are rejected and not compiled at all. On the other hand, compiled well-typed programs are guaranteed to preserve their typing specifications. For example, a program of type \(\mathbb N \times \mathbb Z \to \{\mtt,\mff\}\), for a given input pair \((n,i)\) where~\(n\in \mathbb N\) and~\(i\in \mathbb Z\), is guaranteed to output one of the boolean values (if the computation terminates).

In some situations it might be fairly useful to assume that either the input or output satisfies some further properties, rather than just merely be of a given type. This can be achieved by specifying new types by comprehension on previously defined types. For example, \(T' = \{ x : T \mid P x \}\) specifies a new type \(T'\), which is a subtype (subset) of~\(T\), consisting of only the elements \(x\) of \(T\) for which the program~\(P\) on input~\(x\) outputs value `True'.
However, adding comprehension to a programming language can make typing undecidable. For example, for a computer program on integers \(\mathbb Z \to \mathbb Z\), it is not decidable if it restricts to the type \(\mathbb Z \to T\) where \(T = \{ x : \mathbb Z \mid x \geq 0 \}\).

The typing mechanism of programming languages which happens during compile time (i.e.\ much before it is even executed) is called `static typing.' To resolve the issue with comprehension, some kind of `dynamic typing' (i.e.\ typechecking which happens during the execution of the program) is needed.

The \emph{Blame Calculus}~\cite{wadler2009blame} of Wadler and Findler is a mixed `static' and `dynamic' typing system whose main property is that the statically typed part of a program is guaranteed to be always correct.

The underlying mechanism of the Blame Calculus gives the programmer the ability to \emph{cast} between compatible types. Each such cast is marked with a label and if, during the execution, the cast fails then a blame of the given label is reported.
For example, if a program \(A\) is of type \(S\), the program
\[
    \left< T \Leftarrow S \right>^p A
\]
is of type \(T\) and \(p\) is the blame label for this cast. Then, during the run of this program,~\(A\) calculates some value \(x\) of type \(S\) and then, as part of executing \(\left< T \Leftarrow S \right>^p\), it is checked that \(x\) is also of type \(T\). If this is the case then \(x\) is simply the returned value. On the other hand, if \(x\) fails to be of type~\(T\) then \emph{blame} is allocated to~\(p\), i.e.\ the program ends with the error code \(p\).

For the purposes of casting between function types, for every blame label \(p\) we also introduce its negative version \(\ol p\). However, double negations cancel out, i.e.~\(\ol{\ol p} = p\). This allows to distinguish blames of the input and output casts. For example, for a program~\(A\) of type \(S \to S'\) and a value \(v\) of type \(T\), the cast of \(A\) to type \(T \to T'\) and its application to the value~\(v\), written as
\[
    \big(\left< (T\to T') \Leftarrow (S \to S') \right>^p A\,\big) v
\]
is the program that first casts \(v\) to type \(S\), applies \(A\) to this and then casts the result back to \(T'\). Importantly, the first blame label is \(\ol p\) whereas the second blame label is~\(p\). In other words, the above is computed as:
\[
    \left< T' \Leftarrow S' \right>^p \big(A \, (\left< S \Leftarrow T \right>^{\ol p} v) \big)
\]
We see that if the first cast \(\left< S \Leftarrow T \right>^{\ol p} v\) fails, negative blame \(\ol p\) is reported, which indicates on the issue of the outer environment of this program fragment. On the other hand, if the final cast \(\left< T' \Leftarrow S' \right>^p\) fails, it means that that the returned value of \(A\) does not follow the expected specification. This tension between `external' and `internal' failures (i.e.\ blaming \(\ol p\) and \(p\), respectively) is central to Blame Calculus.

The key result of \cite{wadler2009blame} is that casting from a `more precise' type to a `less precise' type can only give rise to negative blame and, conversely, casting from a `less precise' to a `more precise' type can only give positive blame.
To make this formal the Blame Calculus relies on various notions of `subtyping'. We do not want to go into the intricate details of the definitions of these notions here, the interested reader is advised to study \cite{wadler2009blame} in detail.
For us, it is enough to know that there are two main subtyping relations: \emph{standard subtyping} \(<:\) and \emph{naive subtyping} \(<:_n\). These relations are inductively generated from a few simple rules where the following two demonstrate the main difference between \(<:\) and \(<:_n\).
\[
    \frac{S' <: S \qquad T <: T'}{S \to T <: S' \to T'}
    \qquad
    \qquad
    \frac{S <:_n S' \qquad T <:_n T'}{S \to T <:_n S' \to T'}
\]
The standard subtyping relation is set-up so that \(S <: T\) ensures that any cast from~\(S\) to~\(T\) never causes any blame. On the other hand, \(S <:_n T\) expresses that the type~\(S\) is `more precise' than \(T\).

The fact that these relations are, in fact, orthogonal to each other is demonstrated in terms of two
auxiliary subtyping relations, the \emph{positive} and \emph{negative} subtyping relations, denoted by
\[
    S <:^+ T
    \qtq{and}
    S <:^- T,
\]
respectively.

The positive subtyping relation \(S <:^+ T\) is defined in such a way so that a cast from \(S\) to \(T\) can only receive negative blame (i.e.\ positive blame cannot occur) and, on the other hand, \(S <:^- T\) ensures that a cast from \(S\) to \(T\) cannot receive negative blame.

One of the main technical results of \cite{wadler2009blame} is that standard and naive subtyping relations can be rewritten in terms of these.

\begin{lemma}[Propositions~7 and 8 in \cite{wadler2009blame}]
    \label{l:naive-vs-pm}
    Given types \(S\) and \(T\),
    \begin{enumerate}
        \item \(S <: T\) if and only if \(S <:^+ T\) and \(S <:^- T\).
        \item \(S <:_n T\) if and only if \(S <:^+ T\) and \(T <:^- S\).
    \end{enumerate}
\end{lemma}

This lemma confirms the earlier promise that \(S <: T\) guarantees that no blame arises when casting from \(S\) to \(T\) and, dually, the desired property that \(S\) being `more precise' than \(T\) (i.e.\ \(S <:_n T\)) ensures that casting from \(S\) to \(T\) cannot raise negative blame and, conversely, casting from \(T\) to \(S\) cannot raise negative blame.

The relationship to Belnap's two orders, the information and logic orders, and to the twist construction representation is evident here.
The standard subtyping relation~\(<:\) corresponds to Belnap's information order~\(\ileq\) and the naive subtyping relation~\(<:_n\) corresponds to the logical order~\(\lleq\).
We leave it for the future work to explore how this striking connection can be further exploited.

\section{LVars}
\label{s:lvars}

Let us finish with the last imprint of Belnap--Dunn logic in computer science which is of a very practical nature.
It concerns with \emph{threads} which are the parts of a computer program that run in parallel to each other. They are used to speed up computation by running each of the different threads on a different core of the computer CPU.
Threads often need to communicate with each other and this can be done via shared \emph{variables} (in sense of programming languages), i.e.\ locations in computer memory where some piece of data is stored.

However, variables shared among parallel threads pose many challenges.
In particular, it is not clear what the outcome should be if two threads try to write to the same variable at the same time. In case the two threads are attempting to write different values, how do we decide which value gets stored in the variable at the end?

There are many ways to deal with this problem. Typically, we restrict the access to the shared variables, e.g.\ by only enabling writes to some threads, temporarily locking variables (mutexes, atomic writes), or upon a failed write restarting the thread~(STM).
The disadvantages of these approaches are either reduced efficiency or limited expressivity.

Kuper and Newton~\cite{KuperN13icfp,kuper2015thesis} propose a new type of variables, called \emph{LVars}, to tackle some of these limitations. The main idea is to structure our data in a \emph{bounded join-semilattice} \((S,\ivee,\bot,\top)\), which can be implemented in traditional data structures such as trees, arrays, etc.
Then, every write to a shared variable is computed as a join with its current value, which is guarantees to always increase the stored value.
Also, due to the associativity of \(\ivee\), it does not matter in which order we do the writes.

One way LVar types can be constructed is by taking an arbitrary algebra and freely adding the order. For example, starting from the structure \((\mathbb N, +, \times, 0,1)\) we define the semilattice \(S\) on the set \(\mathbb N \cup \{\bot,\top\}\) by
\[
    (\forall n\in \mathbb N) \ \bot \ileq n \ileq \top
    \qtq{and}
    (\forall n,m\in \mathbb N) \ n \ivee m = \top.
\]
The value \(\bot\) represents the initial value of the variable. If the variable's value is~\(x\) and a thread writes to the variable the value \(n + m\), the new value of the variable will be~\(x \ivee (n+m)\). This means that unless \(x = n+m\), the new value will be \(\top\).

In case we freely add order to the two-variable boolean algebra \(\two\), we obtain the lattice \FOUR{} (from Figure~\ref{fig:four}) in the information order.
Admittedly, LVars are not as closed to the four-valued logic of Belnap as the other three topics we discussed. However, the underlying idea is very much in the spirit of Belnap, that is, an extra order (the information order) needs to be added in order to deal with computational phenomena.
In fact, Belnap \cite{belnap76,belnap77} argues that his logic is suitable for scenarios where the computer obtains data from a variety of (potentially contradicting) sources. In this case, the different sources are the different threads of the computer program.

\bmhead{Acknowledgements}
The author would like to thank Wesley Fussner for useful discussions on the topics of twist products and linear logic, and Umberto Rivieccio and Achim Jung for their useful remarks on an earlier version of this manuscript.
The author is also grateful to Philip Wadler for his talk at the Logic and Semantics Seminar in Cambridge, from March 2023, which inspired the author to look at Blame Calculus.

This project has received funding from the European Union's Horizon 2020 research and innovation programme under the Marie SkÂłodowska-Curie grant agreement No~101111373.


\begin{thebibliography}{37}
\ifx \bisbn   \undefined \def \bisbn  #1{ISBN #1}\fi
\ifx \binits  \undefined \def \binits#1{#1}\fi
\ifx \bauthor  \undefined \def \bauthor#1{#1}\fi
\ifx \batitle  \undefined \def \batitle#1{#1}\fi
\ifx \bjtitle  \undefined \def \bjtitle#1{#1}\fi
\ifx \bvolume  \undefined \def \bvolume#1{\textbf{#1}}\fi
\ifx \byear  \undefined \def \byear#1{#1}\fi
\ifx \bissue  \undefined \def \bissue#1{#1}\fi
\ifx \bfpage  \undefined \def \bfpage#1{#1}\fi
\ifx \blpage  \undefined \def \blpage #1{#1}\fi
\ifx \burl  \undefined \def \burl#1{\textsf{#1}}\fi
\ifx \doiurl  \undefined \def \doiurl#1{\url{https://doi.org/#1}}\fi
\ifx \betal  \undefined \def \betal{\textit{et al.}}\fi
\ifx \binstitute  \undefined \def \binstitute#1{#1}\fi
\ifx \binstitutionaled  \undefined \def \binstitutionaled#1{#1}\fi
\ifx \bctitle  \undefined \def \bctitle#1{#1}\fi
\ifx \beditor  \undefined \def \beditor#1{#1}\fi
\ifx \bpublisher  \undefined \def \bpublisher#1{#1}\fi
\ifx \bbtitle  \undefined \def \bbtitle#1{#1}\fi
\ifx \bedition  \undefined \def \bedition#1{#1}\fi
\ifx \bseriesno  \undefined \def \bseriesno#1{#1}\fi
\ifx \blocation  \undefined \def \blocation#1{#1}\fi
\ifx \bsertitle  \undefined \def \bsertitle#1{#1}\fi
\ifx \bsnm \undefined \def \bsnm#1{#1}\fi
\ifx \bsuffix \undefined \def \bsuffix#1{#1}\fi
\ifx \bparticle \undefined \def \bparticle#1{#1}\fi
\ifx \barticle \undefined \def \barticle#1{#1}\fi
\bibcommenthead
\ifx \bconfdate \undefined \def \bconfdate #1{#1}\fi
\ifx \botherref \undefined \def \botherref #1{#1}\fi
\ifx \url \undefined \def \url#1{\textsf{#1}}\fi
\ifx \bchapter \undefined \def \bchapter#1{#1}\fi
\ifx \bbook \undefined \def \bbook#1{#1}\fi
\ifx \bcomment \undefined \def \bcomment#1{#1}\fi
\ifx \oauthor \undefined \def \oauthor#1{#1}\fi
\ifx \citeauthoryear \undefined \def \citeauthoryear#1{#1}\fi
\ifx \endbibitem  \undefined \def \endbibitem {}\fi
\ifx \bconflocation  \undefined \def \bconflocation#1{#1}\fi
\ifx \arxivurl  \undefined \def \arxivurl#1{\textsf{#1}}\fi
\csname PreBibitemsHook\endcsname

\bibitem[\protect\citeauthoryear{Arieli and Avron}{1996}]{arieli96}
\begin{barticle}
\bauthor{\bsnm{Arieli}, \binits{O.}},
\bauthor{\bsnm{Avron}, \binits{A.}}:
\batitle{Reasoning with logical bilattices}.
\bjtitle{Journal of Logic, Language and Information}
\bvolume{5},
\bfpage{25}--\blpage{63}
(\byear{1996})
\doiurl{10.1007/BF00215626}
\end{barticle}
\endbibitem

\bibitem[\protect\citeauthoryear{Avron}{1988}]{avron1988ll}
\begin{barticle}
\bauthor{\bsnm{Avron}, \binits{A.}}:
\batitle{The semantics and proof theory of linear logic}.
\bjtitle{Theoretical Computer Science}
\bvolume{57}(\bissue{2-3}),
\bfpage{161}--\blpage{184}
(\byear{1988})
\end{barticle}
\endbibitem

\bibitem[\protect\citeauthoryear{Avron}{1996}]{avron1996structure}
\begin{barticle}
\bauthor{\bsnm{Avron}, \binits{A.}}:
\batitle{The structure of interlaced bilattices}.
\bjtitle{Mathematical Structures in Computer Science}
\bvolume{6}(\bissue{3}),
\bfpage{287}--\blpage{299}
(\byear{1996})
\end{barticle}
\endbibitem

\bibitem[\protect\citeauthoryear{Belnap}{1976}]{belnap76}
\begin{bchapter}
\bauthor{\bsnm{Belnap}, \binits{N.D.}}:
\bctitle{How a computer should think}.
In: \beditor{\bsnm{Ryle}, \binits{G.}} (ed.)
\bbtitle{Contemporary Aspects of Philosophy},
pp. \bfpage{30}--\blpage{56}.
\bpublisher{Oriel Press},
\blocation{Boston}
(\byear{1976})
\end{bchapter}
\endbibitem

\bibitem[\protect\citeauthoryear{Belnap}{1977}]{belnap77}
\begin{bchapter}
\bauthor{\bsnm{Belnap}, \binits{N.D.}}:
\bctitle{A useful four-valued logic}.
In: \beditor{\bsnm{Dunn}, \binits{J.M.}},
\beditor{\bsnm{Epstein}, \binits{G.}} (eds.)
\bbtitle{Modern Uses of Multiple-Valued Logic},
pp. \bfpage{8}--\blpage{37}.
\bpublisher{D. Reidel Publishing Company},
\blocation{Dordrecht-Holland}
(\byear{1977})
\end{bchapter}
\endbibitem

\bibitem[\protect\citeauthoryear{Chu}{1978}]{chu1978constructing}
\begin{botherref}
\oauthor{\bsnm{Chu}, \binits{P.-H.}}:
Constructing*-autonomous categories.
Master's thesis,
McGill University
(1978)
\end{botherref}
\endbibitem

\bibitem[\protect\citeauthoryear{Chu}{1979}]{chu1979constructing}
\begin{bchapter}
\bauthor{\bsnm{Chu}, \binits{P.-H.}}:
\bctitle{Constructing $*$-autonomous categories (Appendix)}.
\bbtitle{\emph{In:} $*$-autonomous categories}.
\bsertitle{Lecture Notes in Mathematics},
vol. \bseriesno{752},
pp. \bfpage{103}--\blpage{138}.
\bpublisher{Springer},
\blocation{Berlin}
(\byear{1979})
\end{bchapter}
\endbibitem

\bibitem[\protect\citeauthoryear{Dunn}{1976}]{Dunn1976}
\begin{barticle}
\bauthor{\bsnm{Dunn}, \binits{J.M.}}:
\batitle{Intuitive semantics for first-degree entailments and 'coupled trees'}.
\bjtitle{Philosophical Studies}
\bvolume{29}(\bissue{3}),
\bfpage{149}--\blpage{168}
(\byear{1976})
\doiurl{10.1007/bf00373152}
\end{barticle}
\endbibitem

\bibitem[\protect\citeauthoryear{Escard{\'o}}{2004}]{escardo2004synthetic}
\begin{barticle}
\bauthor{\bsnm{Escard{\'o}}, \binits{M.}}:
\batitle{Synthetic topology: of data types and classical spaces}.
\bjtitle{Electronic Notes in Theoretical Computer Science}
\bvolume{87},
\bfpage{21}--\blpage{156}
(\byear{2004})
\end{barticle}
\endbibitem

\bibitem[\protect\citeauthoryear{Fitting}{1991}]{fitting91bilattices}
\begin{barticle}
\bauthor{\bsnm{Fitting}, \binits{M.}}:
\batitle{Bilattices and the semantics of logic programming}.
\bjtitle{Journal of Logic Programming}
\bvolume{11},
\bfpage{91}--\blpage{116}
(\byear{1991})
\end{barticle}
\endbibitem

\bibitem[\protect\citeauthoryear{Fussner and
  Galatos}{2019}]{fussner2019mingles}
\begin{barticle}
\bauthor{\bsnm{Fussner}, \binits{W.}},
\bauthor{\bsnm{Galatos}, \binits{N.}}:
\batitle{Categories of models of r-mingle}.
\bjtitle{Annals of Pure and Applied Logic}
\bvolume{170}(\bissue{10}),
\bfpage{1188}--\blpage{1242}
(\byear{2019})
\doiurl{10.1016/j.apal.2019.05.003}
\end{barticle}
\endbibitem

\bibitem[\protect\citeauthoryear{Garfield}{2019}]{Garfield2019}
\begin{bbook}
\bauthor{\bsnm{Garfield}, \binits{J.L.}}:
In: \beditor{\bsnm{Omori}, \binits{H.}},
\beditor{\bsnm{Wansing}, \binits{H.}} (eds.)
\bbtitle{Belnap and N{\={a}}g{\={a}}rjuna on How Computers and Sentient Beings
  Should Think: Truth, Trust and the Catuṣkoṭi},
pp. \bfpage{147}--\blpage{153}.
\bpublisher{Springer},
\blocation{Cham}
(\byear{2019}).
\doiurl{10.1007/978-3-030-31136-0_10} .
\burl{https://doi.org/10.1007/978-3-030-31136-0_10}
\end{bbook}
\endbibitem

\bibitem[\protect\citeauthoryear{Ginsberg}{1988}]{ginsberg1988multivalued}
\begin{barticle}
\bauthor{\bsnm{Ginsberg}, \binits{M.L.}}:
\batitle{Multivalued logics: {A} uniform approach to reasoning in artificial
  intelligence}.
\bjtitle{Computational Intelligence}
\bvolume{4}(\bissue{3}),
\bfpage{265}--\blpage{316}
(\byear{1988})
\doiurl{10.1111/j.1467-8640.1988.tb00280.x}
\end{barticle}
\endbibitem

\bibitem[\protect\citeauthoryear{Girard}{1987}]{girard1987linear}
\begin{barticle}
\bauthor{\bsnm{Girard}, \binits{J.-Y.}}:
\batitle{Linear logic}.
\bjtitle{Theoretical computer science}
\bvolume{50}(\bissue{1}),
\bfpage{1}--\blpage{101}
(\byear{1987})
\end{barticle}
\endbibitem

\bibitem[\protect\citeauthoryear{Goubault-Larrecq}{2013}]{GoubaultLarrecq2013}
\begin{bbook}
\bauthor{\bsnm{Goubault-Larrecq}, \binits{J.}}:
\bbtitle{Non-Hausdorff Topology and Domain Theory: Selected Topics in Point-Set
  Topology}.
\bpublisher{Cambridge University Press},
\blocation{UK}
(\byear{2013})
\end{bbook}
\endbibitem

\bibitem[\protect\citeauthoryear{Jakl}{2018}]{jakl2018}
\begin{botherref}
\oauthor{\bsnm{Jakl}, \binits{T.}}:
d-frames as algebraic duals of bitopological spaces.
PhD thesis,
Charles University and University of Birmingham
(2018).
\url{http://etheses.bham.ac.uk/8380/}
\end{botherref}
\endbibitem

\bibitem[\protect\citeauthoryear{Jakl et~al.}{2016}]{jakljungpultr2016}
\begin{barticle}
\bauthor{\bsnm{Jakl}, \binits{T.}},
\bauthor{\bsnm{Jung}, \binits{A.}},
\bauthor{\bsnm{Pultr}, \binits{A.}}:
\batitle{Bitopology and four-valued logic}.
\bjtitle{Electronic Notes in Theoretical Computer Science}
\bvolume{325},
\bfpage{201}--\blpage{219}
(\byear{2016})
\doiurl{10.1016/j.entcs.2016.09.039} .
\bcomment{The Thirty-second Conference on the Mathematical Foundations of
  Programming Semantics (MFPS XXXII)}
\end{barticle}
\endbibitem

\bibitem[\protect\citeauthoryear{Jayatilleke}{1967}]{jayatilleke67indian}
\begin{barticle}
\bauthor{\bsnm{Jayatilleke}, \binits{K.N.}}:
\batitle{The logic of four alternatives}.
\bjtitle{Philosophy East and West}
\bvolume{17}(\bissue{1/4}),
\bfpage{69}--\blpage{83}
(\byear{1967})
\end{barticle}
\endbibitem

\bibitem[\protect\citeauthoryear{Jung and Moshier}{2006}]{jung06}
\begin{botherref}
\oauthor{\bsnm{Jung}, \binits{A.}},
\oauthor{\bsnm{Moshier}, \binits{M.A.}}:
On the bitopological nature of {S}tone duality.
Technical Report CSR-06-13,
School of Computer Science, The University of Birmingham
(2006).
110 pages.
\url{ftp://ftp.cs.bham.ac.uk/pub/tech-reports/2006/CSR-06-13.pdf}
\end{botherref}
\endbibitem

\bibitem[\protect\citeauthoryear{Kalman}{1958}]{kalman1958lattices}
\begin{barticle}
\bauthor{\bsnm{Kalman}, \binits{J.A.}}:
\batitle{Lattices with involution}.
\bjtitle{Transactions of the American Mathematical Society}
\bvolume{87}(\bissue{2}),
\bfpage{485}--\blpage{491}
(\byear{1958})
\end{barticle}
\endbibitem

\bibitem[\protect\citeauthoryear{Kuper}{2015}]{kuper2015thesis}
\begin{botherref}
\oauthor{\bsnm{Kuper}, \binits{L.}}:
Lattice-based data structures for deterministic parallel and distributed
  programming.
PhD thesis,
Indiana University
(2015)
\end{botherref}
\endbibitem

\bibitem[\protect\citeauthoryear{Kuper and Newton}{2013}]{KuperN13icfp}
\begin{bchapter}
\bauthor{\bsnm{Kuper}, \binits{L.}},
\bauthor{\bsnm{Newton}, \binits{R.R.}}:
\bctitle{{LVars}: lattice-based data structures for deterministic parallelism}.
In: \beditor{\bsnm{Grelck}, \binits{C.}},
\beditor{\bsnm{Henglein}, \binits{F.}},
\beditor{\bsnm{Acar}, \binits{U.A.}},
\beditor{\bsnm{Berthold}, \binits{J.}} (eds.)
\bbtitle{Proceedings of the 2nd {ACM} {SIGPLAN} Workshop on Functional
  High-performance Computing, September 25-27, 2013},
pp. \bfpage{71}--\blpage{84}.
\bpublisher{{ACM}},
\blocation{Boston, MA, USA}
(\byear{2013}).
\doiurl{10.1145/2502323.2502326} .
\burl{https://doi.org/10.1145/2502323.2502326}
\end{bchapter}
\endbibitem

\bibitem[\protect\citeauthoryear{Lawson}{2010}]{lawson2011stably}
\begin{barticle}
\bauthor{\bsnm{Lawson}, \binits{J.}}:
\batitle{Stably compact spaces}.
\bjtitle{Mathematical Structures in Computer Science}
\bvolume{21}(\bissue{01}),
\bfpage{125}--\blpage{169}
(\byear{2010})
\doiurl{10.1017/s0960129510000319}
\end{barticle}
\endbibitem

\bibitem[\protect\citeauthoryear{Matsakis and Klock}{2014}]{MatsakisKlock2014}
\begin{bchapter}
\bauthor{\bsnm{Matsakis}, \binits{N.D.}},
\bauthor{\bsnm{Klock}, \binits{F.S.}}:
\bctitle{The rust language}.
In: \bbtitle{Proceedings of the 2014 ACM SIGAda Annual Conference on High
  Integrity Language Technology}.
\bsertitle{HILT '14},
pp. \bfpage{103}--\blpage{104}.
\bpublisher{Association for Computing Machinery},
\blocation{New York, NY, USA}
(\byear{2014}).
\doiurl{10.1145/2663171.2663188} .
\burl{https://doi.org/10.1145/2663171.2663188}
\end{bchapter}
\endbibitem

\bibitem[\protect\citeauthoryear{Mislove}{1998}]{mislove1998topology}
\begin{barticle}
\bauthor{\bsnm{Mislove}, \binits{M.W.}}:
\batitle{Topology, domain theory and theoretical computer science}.
\bjtitle{Topology and its Applications}
\bvolume{89}(\bissue{1-2}),
\bfpage{3}--\blpage{59}
(\byear{1998})
\end{barticle}
\endbibitem

\bibitem[\protect\citeauthoryear{Picado and Pultr}{2012}]{picadopultr2011}
\begin{bbook}
\bauthor{\bsnm{Picado}, \binits{J.}},
\bauthor{\bsnm{Pultr}, \binits{A.}}:
\bbtitle{Frames and Locales}.
\bpublisher{Springer},
\blocation{Basel}
(\byear{2012}).
\doiurl{10.1007/978-3-0348-0154-6} .
\burl{https://doi.org/10.1007/978-3-0348-0154-6}
\end{bbook}
\endbibitem

\bibitem[\protect\citeauthoryear{Rivieccio}{2010}]{rivieccio10}
\begin{botherref}
\oauthor{\bsnm{Rivieccio}, \binits{U.}}:
An algebraic study of bilattice-based logics.
PhD thesis,
University of Barcelona
(2010)
\end{botherref}
\endbibitem

\bibitem[\protect\citeauthoryear{Rivieccio}{2014}]{rivieccio2014implicative}
\begin{barticle}
\bauthor{\bsnm{Rivieccio}, \binits{U.}}:
\batitle{Implicative twist-structures}.
\bjtitle{Algebra universalis}
\bvolume{71}(\bissue{2}),
\bfpage{155}--\blpage{186}
(\byear{2014})
\doiurl{10.1007/s00012-014-0272-5}
\end{barticle}
\endbibitem

\bibitem[\protect\citeauthoryear{Rivieccio and
  Busaniche}{2024}]{RivieccioBusaniche2024conuclei}
\begin{barticle}
\bauthor{\bsnm{Rivieccio}, \binits{U.}},
\bauthor{\bsnm{Busaniche}, \binits{M.}}:
\batitle{Nelson conuclei and nuclei: The twist construction beyond
  involutivity}.
\bjtitle{Studia Logica}
\bvolume{112}(\bissue{5}),
\bfpage{1123}--\blpage{1161}
(\byear{2024})
\doiurl{10.1007/s11225-023-10088-9}
\end{barticle}
\endbibitem

\bibitem[\protect\citeauthoryear{Rivieccio et~al.}{2018}]{rivieccio2020non}
\begin{barticle}
\bauthor{\bsnm{Rivieccio}, \binits{U.}},
\bauthor{\bsnm{Maia}, \binits{P.}},
\bauthor{\bsnm{Jung}, \binits{A.}}:
\batitle{Non-involutive twist-structures}.
\bjtitle{Logic Journal of the IGPL}
\bvolume{28}(\bissue{5}),
\bfpage{973}--\blpage{999}
(\byear{2018})
\doiurl{10.1093/jigpal/jzy070}
\end{barticle}
\endbibitem

\bibitem[\protect\citeauthoryear{Robinson}{1957}]{robinson57indian}
\begin{barticle}
\bauthor{\bsnm{Robinson}, \binits{R.H.}}:
\batitle{Some {L}ogical {A}spects of {N}\={a}g\={a}rjuna's {S}ystem}.
\bjtitle{Philosophy East and West}
\bvolume{6}(\bissue{4}),
\bfpage{291}--\blpage{308}
(\byear{1957})
\end{barticle}
\endbibitem

\bibitem[\protect\citeauthoryear{Scott}{1970}]{scott1970outline}
\begin{bchapter}
\bauthor{\bsnm{Scott}, \binits{D.S.}}:
\bctitle{Outline of a mathematical theory of computation}.
In: \bbtitle{4th Annual Princeton Conference on Information Sciences and
  Systems},
pp. \bfpage{169}--\blpage{176}
(\byear{1970})
\end{bchapter}
\endbibitem

\bibitem[\protect\citeauthoryear{Scott}{1993}]{scott1993}
\begin{barticle}
\bauthor{\bsnm{Scott}, \binits{D.S.}}:
\batitle{A type-theoretical alternative to iswim, cuch, owhy}.
\bjtitle{Theoretical Computer Science}
\bvolume{121}(\bissue{1}),
\bfpage{411}--\blpage{440}
(\byear{1993})
\doiurl{10.1016/0304-3975(93)90095-B}
\end{barticle}
\endbibitem

\bibitem[\protect\citeauthoryear{Shulman}{2018}]{Shulman2018linear}
\begin{botherref}
\oauthor{\bsnm{Shulman}, \binits{M.}}:
Linear logic for constructive mathematics.
Preprint, version 1 from 19 May 2018
(2018).
\url{https://arxiv.org/abs/1805.07518v1}
\end{botherref}
\endbibitem

\bibitem[\protect\citeauthoryear{Shulman}{2022}]{Shulman2022affine}
\begin{barticle}
\bauthor{\bsnm{Shulman}, \binits{M.}}:
\batitle{Affine logic for constructive mathematics}.
\bjtitle{The Bulletin of Symbolic Logic}
\bvolume{28}(\bissue{3}),
\bfpage{327}--\blpage{386}
(\byear{2022})
\doiurl{10.1017/bsl.2022.28}
\end{barticle}
\endbibitem

\bibitem[\protect\citeauthoryear{Tsinakis and Wille}{2006}]{tsinakis2006}
\begin{barticle}
\bauthor{\bsnm{Tsinakis}, \binits{C.}},
\bauthor{\bsnm{Wille}, \binits{A.M.}}:
\batitle{Minimal varieties of involutive residuated lattices}.
\bjtitle{Studia Logica}
\bvolume{83}(\bissue{1}),
\bfpage{407}--\blpage{423}
(\byear{2006})
\doiurl{10.1007/s11225-006-8311-7}
\end{barticle}
\endbibitem

\bibitem[\protect\citeauthoryear{Wadler and Findler}{2009}]{wadler2009blame}
\begin{bchapter}
\bauthor{\bsnm{Wadler}, \binits{P.}},
\bauthor{\bsnm{Findler}, \binits{R.B.}}:
\bctitle{Well-typed programs can’t be blamed}.
In: \bbtitle{European Symposium on Programming},
pp. \bfpage{1}--\blpage{16}
(\byear{2009}).
\bcomment{Springer}
\end{bchapter}
\endbibitem

\end{thebibliography}

\end{document}